\documentclass[journal]{IEEEtran}
\usepackage{amsfonts,amsmath,latexsym,amssymb}
\usepackage{graphicx}
\usepackage{subfigure}
\usepackage{amsmath}
\usepackage{theorem}
\usepackage{multirow}
\usepackage{amssymb}
\usepackage{array}
\usepackage{cite}
\usepackage{color}
\newtheorem{Theorem}{Theorem}
\newtheorem{lemma}{\textbf{Lemma}}

\newtheorem{DD}{Definition}
\usepackage[]{algorithm2e}
\usepackage{fancyhdr}

\begin{document}
\title{Partial Third-Party Information Exchange with Network Coding}
\author{Xiumin Wang,~
        Chau Yuen
\thanks{Xiumin Wang is with School of Computer and Information, Hefei University of Technology, China. E-mail: wxiumin@hfut.edu.cn.}
\thanks{Chau Yuen is with Singapore University of Technology and Design, Singapore. Email: yuenchau@sutd.edu.sg.}
\thanks{This research is partly supported by the International Design Center (grant no. IDG31100102 \& IDD11100101). It is also supported in part by the Fundamental Research
Funds for the Central Universities (No. 2012HGBZ0640).}
}

\maketitle

\pagenumbering{arabic}

\begin{abstract}
In this paper, we consider the problem of exchanging channel state
information in a wireless network such that a subset of the clients
can obtain the complete channel state information of all the links
in the network. We first derive the minimum number of required
transmissions for such partial third-party information exchange
problem. We then design an optimal transmission scheme by
determining the number of packets that each client should send, and
designing a deterministic encoding strategy such that the subset of
clients can acquire complete channel state information of the
network with minimal number of transmissions. Numerical
results show that network coding can efficiently reduce the number
of transmissions, even with only pairwise encoding.

\end{abstract}
%\begin{IEEEkeywords}
%network coding, cooperative data exchange, channel state information.
%\end{IEEEkeywords}

\section{Introduction}
To increase throughput and network efficiency, it is always
beneficial for clients in a wireless network to have a certain level
of global knowledge of the network, for example link loss
probability that is related to the network connectivity, or channel
state information (CSI) that is related to the channel quality.
Generally, such information, e.g. link loss probability and CSI, on
a link $(i,j)$ is regarded as a local and common information between
two connected nodes $i$ and $j$, and it is unknown to a {\em
third-party} node, e.g., the node $k\neq i,j$. Thus, the problem of
letting third-party nodes to know the information that is local to
the other nodes is an important problem for network design
\cite{Love2007,Wang2012b}.

Recently, cooperative data exchange
\cite{ElRouayheb2010,Sprintson2010,Tajbakhsh2011,Milosavljevic2011}
with network coding  \cite{Ahlswede2000,Katti2006} has become a
promising approach for efficient data communication. In cooperative
data exchange, the clients in a network exchange the packets via a
lossless common/broadcast channel. Inspired by cooperative data
exchange, the works in \cite{Love2007}\cite{Wang2012b} proposed a
network coding based {\em third-party information exchange} where
the clients exchange their local CSI through a lossless
common/broadcast channel such that each client finally gets the
complete CSI of the whole network. Specifically, the work in
\cite{Love2007} aims to minimize the total number of transmissions,
while the work in \cite{Wang2012b} tries to minimize the total
transmission cost required to complete the information exchange.
However, both of these works assume that all the clients in the
network need to get the complete CSI. In a practical system, there
could be only a subset of the clients, e.g. only the one with
information to transmit or receive, need to get the complete CSI at
the same time in most cases.

Hence, different from the previous works
\cite{Love2007}\cite{Wang2012b}, we study a practical scenario where
only a subset of the clients in the network need to acquire the
complete CSI in this paper. To simplify the presentation, such
third-party information exchange for only a subset of clients is
denoted as {\em partial third-party information exchange}. We aim to
propose a network coding based solution to minimize the total number
of transmissions required during partial third-party information
exchange. The main contributions of this paper can be summarized as
follows:
\begin{itemize}
\item We derive the minimum number of transmissions for such partial third-party information exchange.
\item We propose an optimal transmission scheme, which determines the number of packets that each client should send so as to achieve the minimal number of transmissions.
\item We design a deterministic encoding strategy to make sure that
with the proposed transmission scheme, the subset of clients that
require the complete CSI can successfully decode/obtain the full
information.
\item Numerical results show that the proposed transmission scheme and encoding strategy can efficiently reduce the number of transmissions. 
\end{itemize}

%The rest of the papers are organized as follows. The problem is
%defined in Section~\ref{Sec.definition}. Section~\ref{Sec.lower}
%derives the optimal minimum number of transmissions. In Section~\ref{Sec.optimal}, we propose an optimal transmission scheme and a deterministic code design. We conclude the
%paper in Section~\ref{Sec.conclusion}.

\section{Problem Description}\label{Sec.definition}
Consider a network with $N$ clients in $C=\{c_1,c_2,\cdots, c_N\}$.
Let $x_{i,j}$ represent the link loss probability or CSI of the
link between clients $c_i$ and $c_j$, where each client $c_i$ only
knows the local information initially, i.e., client $c_i$ only holds
the CSI in $X_i=\{x_{i,j}|\forall j\in\{1,2,\cdots,N\}\setminus
\{i\}\}$. We assume that the links are symmetric, i.e.,
$x_{i,j}=x_{j,i}$ for $\forall i,j$, so for every two clients $c_i$
and $c_{j}$, they hold one common CSI $x_{i,j}$. Thus, the set of
all the CSI in the network is $X=\{x_{i,j} | 1 \leq i,j \leq N,
i \neq j \}$, and the total number of the packets is
$|X|=\frac{N(N-1)}{2}$.

Instead of letting all the clients get the complete CSI \cite{Love2007}\cite{Wang2012b},
in this paper, we consider the case that only a subset of clients, $C'\subseteq
C$, need to get the complete CSI in $X$. Without loss of
generality, we assume that the first $k$ clients in $C$ want all the
CSI in $X$, i.e. $C'=\{c_1,c_2,\cdots, c_k\}$, $k\leq N$.
We also use $\overline{X_i}$ to denote the
set of ``wanted" packets by client $c_i\in C'$, i.e.,
$\overline{X_i}=X\backslash X_i\subseteq X$.

As in \cite{Love2007}\cite{Wang2012b}, there is a lossless
common/broadcast channel for clients to exchange information. Let
$y_i$ be the number of packets that client $c_i$ should send. Then, the
total number of transmissions sent by all the clients in $C$ (notice
that although only a subset of clients $C'$ requires the full
information, all clients in $C$ should participate in sharing their
information) can thus be written as
{\small\begin{eqnarray}\label{objective.cost} Y=\sum_{i=1}^Ny_i
\end{eqnarray}}

Recent works
\cite{Love2007,Wang2012b,ElRouayheb2010,Sprintson2010,Tajbakhsh2011,Milosavljevic2011}
show that network coding can efficiently save the number of
transmissions for data exchange problem. Thus, after determining the
number of transmissions that each client should send, we design an
encoding strategy based on network coding \cite{Ahlswede2000}, where
a linear encoded packet will be generated based on the packets that
the sender initially has over a finite field.

In the following sections, we will derive the minimum number of transmissions required for partial third-party information exchange in Section~\ref{Sec.lower}. Then, in Section~\ref{Sec.optimal}, we propose an optimal transmission scheme, which can achieve the minimum number of transmissions. Based on the proposed transmission scheme, in Section~\ref{Sec.encoding1}, we design a deterministic encoding strategy to make sure that each client in $C'$ can successfully decode/obtain the complete CSI. We compare the performance in Section~\ref{Sec.comp}.

%In this paper, we aim to propose an optimal transmission scheme and design a corresponding encoding strategy, such that the total number of transmissions required to send is minimum while each client in $C'\subseteq C$ can successfully decode/obtain the complete CSI in $X$.

\section{The Minimum Number of Transmissions}\label{Sec.lower}
In this section, we theoretically derive the minimum number of transmissions for the partial third-party information exchange problem.

We use $Y_{min}$ to denote the minimum number of transmissions
required for the partial third-party information exchange problem.
We have the following lemma.

\begin{lemma}\label{lemma.low}
The minimum number of transmissions required for the partial
third-party information exchange problem is lower bounded as
{\small\begin{eqnarray}\label{low} Y_{min}\geq
\frac{(N-1)(N-2)}{2}+\lceil{\frac{k-2}{2}}\rceil
\end{eqnarray}}
\end{lemma}
\begin{proof}
Since the number of packets required by each client $c_i\in C'$ is
$|\overline{X_i}|=\frac{(N-1)(N-2)}{2}$, the number of packets
received by client $c_i$ should be at least $\frac{(N-1)(N-2)}{2}$,
otherwise, $c_i$ cannot get the complete information. In other
words, {\small\begin{eqnarray}
\sum_{i'\in\{1,\cdots,N\}\backslash\{i\}}y_{i'}\geq
\frac{(N-1)(N-2)}{2}
\end{eqnarray}}

By considering all the clients in $C'$, we have
{\small\begin{eqnarray}
\sum_{i=1}^k\sum_{i'\in\{1,\cdots,N\}\backslash\{i\}}y_{i'}&\geq& \frac{(N-1)(N-2)k}{2}\notag\\
k\sum_{i'=1}^Ny_{i'}-\sum_{i'=1}^{k}y_{i'}&\geq& \frac{(N-1)(N-2)k}{2}
\end{eqnarray}}
That is
{\small\begin{eqnarray}\label{ineq.total}
\sum_{i'=1}^Ny_{i'}\geq \frac{(N-1)(N-2)}{2}+\frac{\sum_{i'=1}^{k}y_{i'}}{k}
\end{eqnarray}}

According to \cite{Wang2012b}, for each client $c_i\in C'$, the
packets received from other $k-1$ clients in $C'\backslash\{c_i\}$
should satisfy
{\small\begin{eqnarray} \sum_{c_{i'}\in
C'\backslash\{c_i\}}y_{i'}\geq \binom{k-1}{2}
\end{eqnarray}}
By considering all the clients in $C'$, we have
{\small \begin{eqnarray}
\sum_{i=1}^k\sum_{i'\in\{1,\cdots,k\}\backslash\{i\}}y_{i'}&\geq & \frac{k(k-1)(k-2)}{2}\notag\\
(k-1)\sum_{i'=1}^ky_{i'}&\geq & \frac{k(k-1)(k-2)}{2}
\end{eqnarray}}
That is
{\small\begin{eqnarray}\label{ineq.part}
\sum_{i'=1}^ky_{i'}\geq  \lceil{\frac{k(k-2)}{2}}\rceil
\end{eqnarray}}

According to Eq.~(\ref{ineq.total}) and Eq.~(\ref{ineq.part}), we
can obtain that
{\small\begin{eqnarray}
\sum_{i'=1}^Ny_{i'}\geq \frac{(N-1)(N-2)}{2}+
\lceil{\frac{k-2}{2}}\rceil
\end{eqnarray}}

We thus proved Lemma~\ref{lemma.low}.
\end{proof}

We can also get the following Theorem.
\begin{Theorem}\label{optimal.number}
The optimal number of transmissions required for partial
third-party information exchange problem that minimizes the number of
transmissions is
{\small\begin{eqnarray}\label{equal} Y^{opt}_{min}=
\frac{(N-1)(N-2)}{2}+ \lceil{\frac{k-2}{2}}\rceil
\end{eqnarray}}
\end{Theorem}
\begin{proof}
According to Lemma~\ref{lemma.low}, we know that the lower bound of
the minimum number of transmissions required is
$\frac{(N-1)(N-2)}{2}+ \lceil{\frac{k-2}{2}}\rceil$. Thus, we get
the optimal minimum number of transmissions, as denoted in
Eq.~(\ref{equal}).
\end{proof}

In the following sections, we will show that the above lower bound can be
achieved with an optimal transmission scheme and a deterministic encoding strategy.

\section{An Optimal Transmission Scheme}\label{Sec.optimal}
In this section, we first propose a feasible transmission scheme. We
then prove that the proposed transmission scheme can achieve the
minimum number of transmissions required for partial third-party information exchange problem.

\subsection{A Feasible Transmission Scheme}\label{Sec.optimal.trans}
We first describe the transmission scheme, which determines the
number of packets that each client should send.
\begin{DD}\label{DD.trans}
{\bf The transmission scheme}: the number of packets sent by each client is
{\small\begin{eqnarray}\label{scheme}
y_{i}=\left\{
\begin{aligned}
&\lceil{\frac{k}{2}-1}\rceil,\mbox{ if }1\leq i< k\\
&\frac{k}{2}-1, \mbox{ if i=k and k is even}\\
&0, \mbox{ if i=k and k is odd}\\
&N+k-i-1,\mbox{ if }k+1\leq i\leq N
\end{aligned}
\right.
\end{eqnarray}}
\end{DD}

We can get the following lemma.
\begin{lemma}\label{lemma.feasible}
The transmission scheme defined in Definition~\ref{DD.trans} is a
feasible solution for the partial third-party information exchange problem.
\end{lemma}
\begin{proof}
We prove the above lemma by only considering the case when $k$ is
even. The case when $k$ is odd can be proved in a similar way.

According to \cite{Wang2012b}, there exists a feasible code design
to make sure the client $c_i\in C'$ can successfully decode/obtain
the complete information, if and only if the total number of packets
received from any other $l$ clients in $C\backslash \{c_i\}$ is at
least $\binom{l}{2}$. In other words, the feasible solution of the
partial third-party information exchange requires that for any
$c_i\in C'$, {\small\begin{eqnarray}\label{condition}
\sum_{t=1}^{l}y_{i_t}\geq \binom{l}{2}, \mbox{ for } \forall
\{c_{i_1},c_{i_2},\cdots,c_{i_l}\}\subseteq C\backslash \{c_i\}
\end{eqnarray}}

According to Eq.~(\ref{scheme}), for $\forall
\{c_{i_1},c_{i_2},\cdots,c_{i_l}\}\subseteq C\backslash \{c_i\}$, when $l<k$, we
have
{\small\begin{eqnarray} \sum_{t=1}^ly_{i_t}\geq
(\frac{k}{2}-1)l \geq \frac{l(l-1)}{2}=\binom{l}{2}
\end{eqnarray}}

When $l\geq k$, we have {\small\begin{eqnarray}
\sum_{t=1}^ly_{i_t}&\geq& \sum_{j\in C'\backslash \{i\}} y_j+\sum_{j=N-l+k}^Ny_j\notag \\
%&=& (\frac{k}{2}-1)(k-1)+\frac{(l-k+1)(l+k-2)}{2}\notag\\
&=&\frac{l(l-1)}{2}=\binom{l}{2}
\end{eqnarray}}

We thus prove that the transmission scheme
in Definition~\ref{DD.trans} is a feasible transmission
scheme, i.e., there exists a feasible code design to make sure with
the above transmission scheme, the clients in $C'$ can obtain their
``wanted" packets.
\end{proof}

\subsection{Performance Analysis}
We now prove that the proposed transmission scheme can achieve the
minimum number of transmissions for the partial third-party
information exchange problem as specified in Theorem 1. To avoid
confusion, we use $Y_p$ to denote the number of transmissions
required by the proposed transmission scheme. According to
Definition~\ref{DD.trans}, we can have the following lemma.
\begin{lemma}\label{lemma2}
The number of transmissions required with the transmission scheme defined in Definition~\ref{DD.trans} is
{\small\begin{eqnarray}\label{upp}
Y_{p}=\frac{(N-1)(N-2)}{2}+\lceil{\frac{k-2}{2}}\rceil
\end{eqnarray}}
\end{lemma}
\begin{proof}
We analyze the number of transmissions required with the proposed
transmission scheme by considering two cases: 1) $k$ is even and 2)
$k$ is odd.

Case 1 ($k$ is even): According to
Eq.~(\ref{scheme}), the total number of transmissions sent by all
the clients can be expressed as
{\small\begin{eqnarray}
&&\sum_{i=1}^k(\frac{k}{2}-1)+\sum_{i=k+1}^N(N+k-i-1)\notag\\
%&=&(\frac{k}{2}-1)k+\frac{(N-k)(N+k-3)}{2}\notag\\
&=&\frac{(N-1)(N-2)}{2}+\frac{k-2}{2}
\end{eqnarray}}
Case 2 ($k$ is odd): According to Eq.~(\ref{scheme}), the total
number of transmissions required for this case is
{\small\begin{eqnarray}
&&\sum_{i=1}^{k-1}\lceil\frac{k}{2}-1\rceil+\sum_{i=k+1}^N(N+k-i-1)\notag\\
&=&\frac{(k-1)^2}{2}+\frac{(N-k)(N+k-3)}{2}\notag\\
%&=&\frac{(N-1)(N-2)}{2}+\frac{k-1}{2}\notag\\
&=&\frac{(N-1)(N-2)}{2}+\lceil{\frac{k-2}{2}}\rceil
\end{eqnarray}}
Thus, the number of transmissions required is
{\small\begin{eqnarray}
Y_{p}=\frac{(N-1)(N-2)}{2}+\lceil{\frac{k-2}{2}}\rceil\notag
\end{eqnarray}}
\end{proof}

We then get the following Theorem.
\begin{Theorem}\label{optimal.theorem}
The proposed transmission scheme achieves the optimal solution of the partial third-party information exchange problem.
\end{Theorem}
\begin{proof}
As the proposed feasible transmission scheme achieves the minimum number of transmissions in Eq.(\ref{equal}), the derived bound is thus reachable and the proposed transmission scheme achieves the optimal solution.
\end{proof}

\section{A Deterministic Network Code Design}\label{Sec.encoding1}

The above transmission scheme only gives the number of transmissions to be sent by each client. In this section, we will design a deterministic encoding strategy to decide the encoded packets that each client should send, so as to make sure that
with the above transmission scheme, each client in $C'$ can
successfully decode/obtain all its ``wanted" packets.

\begin{DD}\label{DD.encoding}
Each client encodes the packets according to the following rules, where the number of packets sent by each client is determined by
Definition~\ref{DD.trans}.

(1) For client $c_i$, where $1\leq i\leq k$, the $j$-th packet sent
by $c_i$ is
\begin{eqnarray}\label{encode.1}
x_{i,i\%k+1}\oplus x_{i,(i+j)\%k+1}
\end{eqnarray}

(2) For client $c_i$, where $k<i\leq N$,
\begin{itemize}
\item if $j<k$, the $j$-th packet sent by $c_i$ is
\begin{eqnarray}\label{encode.2}
x_{1,i}\oplus x_{1+j,i}
\end{eqnarray}
\item if $j\geq k$, the $j$-th packet sent by $c_i$ is
\begin{eqnarray}\label{encode.3}
x_{i,i+j-k+1}
\end{eqnarray}
\end{itemize}
\end{DD}

Note that the proposed code design is a simple pairwise coding (i.e.
by encoding at most two packets only), which can be implemented
easily with XOR operation.

Consider an example with $6$ clients, where the first $3$ clients
want to get the complete CSI. According to the above definitions,
the transmission scheme and the encoded packets sent by each client
are shown in Table~\ref{example}.

\begin{table}[ht]
\caption{The transmission scheme and code design for $N=6,k=3$} % title of Table
\centering % used for centering table
\begin{tabular}{|l|l|l|l|}
\hline
& Initial Information & $y_i$ & Code design \\ \hline
$c_1$& $x_{1,2},x_{1,3},x_{1,4},x_{1,5},x_{1,6}$ & 1 & $x_{1,2}\oplus x_{1,3}$\\
\hline
$c_2$& $x_{1,2},x_{2,3},x_{2,4},x_{2,5},x_{2,6}$ & 1 & $x_{2,3}\oplus x_{2,1}$\\
\hline
$c_3$& $x_{1,3},x_{2,3},x_{3,4},x_{3,5},x_{3,6}$ & 0 & \\
\hline
\multirow{2}{*}{$c_4$} &
\multirow{2}{*}{$x_{1,4},x_{2,4},x_{3,4},x_{4,5},x_{4,6}$}&
\multirow{2}{*}{4} & $x_{1,4}\oplus x_{2,4}$\mbox{, }$x_{1,4}\oplus x_{3,4}$\\
& & &$x_{4,5}$\mbox{, }$x_{4,6}$ \\
%\hline
%\multirow{4}{*}{$c_4$} &
%\multirow{4}{*}{$x_{1,4},x_{2,4},x_{3,4},x_{4,5},x_{4,6}$}&
%\multirow{4}{*}{4} & $x_{1,4}\oplus x_{2,4}$\\
% & &  & $x_{1,4}\oplus x_{3,4}$ \\
% & &  & $x_{4,5}$ \\
% & &  & $x_{4,6}$ \\
\hline
\multirow{2}{*}{$c_5$} &
\multirow{2}{*}{$x_{1,5},x_{2,5},x_{3,5},x_{4,5},x_{5,6}$}&
\multirow{2}{*}{3} & $x_{1,5}\oplus x_{2,5}$\mbox{, }$x_{1,5}\oplus x_{3,5}$ \\
 & &  & $x_{5,6}$ \\
\hline
\multirow{1}{*}{$c_6$} &
\multirow{1}{*}{$x_{1,6},x_{2,6},x_{3,6},x_{4,6},x_{5,6}$}&
\multirow{1}{*}{2} & $x_{1,6}\oplus x_{2,6}$\mbox{, }$x_{1,6}\oplus x_{3,6}$ \\
\hline
\end{tabular}\label{example}
\end{table}

We can also prove the following lemma.
\begin{lemma}\label{lemma1}
With the code design in Definition~\ref{DD.encoding} and the
transmission scheme in Definition~\ref{DD.trans}, every client in
$C'$ can successfully decode and obtain the complete CSI in $X$.
\end{lemma}
\begin{proof}
We first prove that every packet $x_{i,j}\in X$ is encoded in at
least one transmitted packet.
\begin{itemize}
\item When $i,j\leq k$, packet $x_{i,j}$ must be encoded in at least one
transmitted packet from client $c_i$ or $c_j$, similar to data
exchange among the $k$ clients in $C'$ \cite{Wang2012b,Love2007}.
\item When $i\leq k$ and $j>k$ (or $i> k$ and $j\leq k$), packet $x_{i,j}$
must be encoded in at least one packet sent by client $c_j$ (or
client $c_i$), according to Eq.~(\ref{encode.2}).
\item When $i,j\geq k$, packet $x_{i,j}$ must be sent by client
$c_{min\{i,j\}}$, according to Eq.~(\ref{encode.3}).
\end{itemize}
Thus, every packet in $X$ will be encoded in at least one
transmitted packet.

We now check the decoding process of the clients in $C'$. We first
consider client $c_1$ as follows:
\begin{itemize}
\item For the packets sent by client $c_i$, where $1<i\leq k$, $c_1$ must
be able to decode them, as this process is similar to the data
exchange among the $k$ clients in $C'$ \cite{Wang2012b,Love2007}.
\item According to Eq.~(\ref{encode.2}), $c_1$ also can decode the first
$k-1$ packets sent by $c_i$, where $i>k$, as packet $x_{1,i}$ is
participated in each of these packets and $x_{1,i}$ is initially
available at $c_1$. In addition, as the other $N-i$ packets sent by
$c_i$ are original ones, $c_1$ can obtain them directly. In other
words, $c_1$ can decode all the packets sent by any client $c_i$,
where $i>k$.
\end{itemize}
As all the packets in $X$ are participated in the packets sent by
the clients and $c_1$ can decode all the packets sent by the
clients, $c_1$ can thus decode/obtain all the CSI in $X$.

We then check the decoding process of client $c_i$, where $1<i\leq
k$. Similar to the decoding process of $c_1$, $c_i$ must be also
able to decode all the packets sent by $c_{i'}$ where $1\leq i'\neq
i\leq k$. In addition, according to Eq.~(\ref{encode.2}), the set of
the first $k-1$ packets sent by client $c_{i'}$, where $i'>k$, is
$\{x_{1,i'}\oplus x_{2,i'},x_{1,i'}\oplus x_{3,i'},\cdots,
x_{1,i'}\oplus x_{k,i'}\}$. In other words, packet $x_{i,i'}$ must
be encoded in one packet sent by $c_{i'}$ where $i\leq k,i'>k$.
Thus, client $c_i$ can decode all the first $k-1$ packets sent by
$c_{i'}$ where $i'>k$. Finally, as the last $N-i$ packets sent by
$c_{i'}$ are original packets, $c_i$ thus can get them directly.
That is, $c_i$ successfully decodes/obtains all the packets sent by
the other clients. As all the packets are participated in the
packets sent by the clients, $c_i$ gets the complete CSI in $X$.

To summarize, with the proposed transmission scheme and
encoding strategy, all the clients
in $C'$ can successfully obtain all the CSI in $X$, which
thus proved Lemma~\ref{lemma1}.
\end{proof}

Still considering the example in Table~\ref{example}, we can easily
verify that after receiving all the packets sent from other clients,
$c_1,c_2$ and $c_3$ can successfully decode their ``wanted" CSI
in $\overline{X_1},\overline{X_2}$ and $\overline{X_3}$
respectively.

Note that the proposed transmission scheme and the encoding strategy
can be implemented in a distributed manner. They only need to know the sequence of the clients and the indices set of the clients in $C'$.

\section{Performance Comparison}\label{Sec.comp}
We now compare the minimum number of transmissions required with and without network coding for various values of $k\leq N$ and $N=4,7,12,15$. As shown in
Table~\ref{table}, we can see that the number
of transmissions with network coding is much less than that without
network coding. Without network coding, the number of transmissions
required for $k\geq 3$ and $k=2$ are $|X|=\frac{N(N-1)}{2}$ and
$|X|-1$ respectively (because when $k\geq 3$, each packet in $X$ is
required by at least one client in $C'$; while when $k=2$, the two
clients in $C'$ must share one common packet). It can also be
verified easily that our result in Theorem
\ref{optimal.theorem} includes \cite{Love2007} as a special case,
where \cite{Love2007} considers to minimize the total number of
transmissions only when $k=N$. However, the encoding scheme proposed in this paper is totally different from \cite{Love2007,Wang2012b} due to different problem setting.

\begin{table}[ht]
\caption{The minimum number of transmissions with and without network coding} % title of Table
\centering % used for centering table
\begin{tabular}{|c|c|c|c|c|c|c|c|} % centered columns (4 columns)
\hline
 &\multicolumn{5}{|c|}{with network coding (NC)}& \multicolumn{2}{|c|}{without NC} \\\hline
 & {\em k=2} & {\em k=4} & {\em k=7} & {\em k=10} & {\em k=15} &{\em k=2} & {\em $3\leq k\leq N$} \\ % inserts table
%heading
\hline % inserts single horizontal line
{\em N=4} & 3 & 4 &  NA & NA &NA   &5  &6 \\ \hline
{\em N=7} & 15 & 16 &  18 & NA &NA &20 &21\\ \hline
%{\em N=10} & 36 & 37 &  39& 40 &NA &44 &45\\ \hline
{\em N=12} & 55 & 56 &  58 & 59 &NA&65 &66\\ \hline
{\em N=15}& 91 & 92 &  94 & 95 &98 &104&105\\ \hline
\end{tabular}
\label{table} % is used to refer this table in the text
\end{table}

\vspace{-0.05in}
\section{Conclusion}\label{Sec.conclusion}
In this paper, we aim to design a network coded transmission scheme
to minimize the total number of transmissions required for partial
third-party information exchange problem. We first derive the
minimum number of required transmissions for the partial third-party
information exchange. Then, we design an optimal transmission scheme
to determine the number of packets that each client should send so
as to achieve the optimal minimal number of transmissions. Finally,
a simple deterministic encoding strategy, based only on XOR
operation, is designed to make sure that with the proposed optimal
transmission scheme, all the clients that require the complete
information can eventually decode/obtain their ``wanted" packets.

\end{document}